\documentclass[a4paper,UKenglish,cleveref, autoref, thm-restate]{lipics-v2021}
\pdfoutput=1 
\hideLIPIcs 

\usepackage[ruled]{algorithm}
\usepackage{algpseudocode}
\newcommand{\tl}{\triangleleft}
\newcommand{\tle}{\trianglelefteq}
\newcommand{\tr}{\triangleright}
\newcommand{\tre}{\trianglerighteq}
\title{The worst-case complexity of symmetric strategy improvement}

\author{Tom {van Dijk}}{Formal Methods and Tools, University of Twente, The Netherlands \and \url{https://www.tvandijk.nl/} }{
t.vandijk@utwente.nl}{}{}

\author{Georg Loho}{Discrete Mathematics and Mathematical Programming, University of Twente, The Netherlands \and \url{https://lohomath.github.io/}}{g.loho@utwente.nl}{}{}

\author{Matthew T. Maat}{Discrete Mathematics and Mathematical Programming, University of Twente, The Netherlands \and \url{https://people.utwente.nl/m.t.maat}}{m.t.maat@utwente.nl}{}{}

\authorrunning{T. van Dijk, G. Loho and M.\,T. Maat } 

\Copyright{Tom van Dijk, Georg Loho and Matthew T. Maat} 

\ccsdesc{Theory of computation~Logic and verification}
\ccsdesc{Mathematics of computing~Discrete mathematics}

\keywords{Parity game, Mean payoff game, Symmetric strategy improvement, Strategy improvement, Worst-case complexity} 

\category{} 

\relatedversion{} 

\acknowledgements{}

\nolinenumbers

\begin{document}

\maketitle
\begin{abstract}
Symmetric strategy improvement is an algorithm introduced by Schewe et al.~(ICALP 2015) that can be used to solve two-player games on directed graphs such as parity games and mean payoff games. In contrast to the usual well-known strategy improvement algorithm, it iterates over strategies of both players simultaneously. The symmetric version solves the known worst-case examples for strategy improvement quickly, however its worst-case complexity remained open. 

We present a class of worst-case examples for symmetric strategy improvement on which this symmetric version also takes exponentially many steps.  
Remarkably, our examples exhibit this behaviour for any choice of improvement rule, which is in contrast to classical strategy improvement where hard instances are usually hand-crafted for a specific improvement rule. 
We present a generalized version of symmetric strategy iteration depending less rigidly on the interplay of the strategies of both players. However, it turns out it has the same shortcomings. 
\end{abstract}

\section{Introduction}
We study certain classes of infinite turn-based games on directed graphs between two players, also called infinitary payoff games, which includes parity games and discounted/mean payoff games. 
These games are interesting from an algorithmic perspective and from the viewpoint of complexity theory.

First, there are various problems that relate to solving these games. Solving parity games is important for formal verification and synthesis of programs, as many properties of programs are naturally specified by means of fixed points; parity games capture the expressive power of nested least and greatest fixed point operators. 
In particular, there are linear reductions between parity games and the model checking problem of the modal $\mu$-calculus~\cite{DBLP:journals/tcs/EmersonJS01,DBLP:conf/stacs/Walukiewicz96}. 
Solving mean payoff games is equivalent to problems like
solving energy games~\cite{BrimChaloupkaDoyenGentiliniRaskin:2011},
deciding feasibility in tropical linear programming~\cite{Akian2012TropicalGames}, 
scheduling with AND/OR precedence constraints~\cite{MoehringSkutellaStork:2004}, and
the max-atoms problem~\cite{DBLP:journals/mst/BodirskyM18}. 

Another notable aspect of these games is their complexity status. It is known that there is a polynomial-time reduction from parity games to mean payoff games, and from mean to discounted payoff games.
Many classes of these games are known to be contained in the intersection of NP and coNP \cite{CONDON1992203, Zwick1996TheGraphs}, and parity games and mean payoff games have been shown to even lie in the intersection of UP and coUP \cite{JURDZINSKI1998119}. However, the question whether there exists a polynomial-time algorithm for any of these games has been open for decades. 

\subparagraph{Related work} 
Many algorithms have been proposed for solving parity games and mean payoff games with the main algorithm classes being value iteration~\cite{Jurdzinski2000SmallGames,DorfmanKaplanZwick:2019,FijalkowGawrychowskiOhlmann:2020}, strategy improvement~\cite{Jurdzinski2000AGames,Bjorklund2007AGames} and attractor-based algorithms~\cite{DBLP:conf/cav/Dijk18,Zielonka1998InfiniteTrees}, where we list only a small part of the many papers. 
For most of these algorithms there are instances which take exponentially many steps; these are usually simple for value iteration, while work by Van Dijk~\cite{DBLP:journals/corr/abs-1807-10210} demonstrates an exponential lower bound to many attractor-based algorithms. 
Recently, it has been shown that parity games can actually be solved in quasi-polynomial time: after the breakthrough in~\cite{Calude2017DecidingTime}, several other quasi-polynomial algorithms have been found, including \cite{DBLP:journals/fcomp/DellErbaS22, DBLP:journals/sttt/FearnleyJKSSW19,Jurdzinski2017SuccinctGames,Parys2019ParityTime}.
However, most of these approaches are likely to be inherently superpolynomial as demonstrated in~\cite{conf/soda/CzerwinskiDFJLP19}.

Strategy improvement \cite{Bjorklund2007AGames, Fearnley2017EfficientGames, Jurdzinski2000AGames, Ohlmann:2021, Puri1995TheorySystems, DBLP:conf/csl/Schewe08} (also called strategy iteration or policy iteration) is considered to be a viable candidate for a polynomial-time algorithm for many classes of infinitary payoff games due to its inherent combinatorial nature. This method evaluates strategies by means of a function on the nodes in the graph called the \emph{valuation}. It then iteratively makes changes to a strategy, improving the valuation, until an optimal strategy is found. When there are multiple options for improvements, the choice is made by a so-called \emph{improvement rule}. 
There are a few valuations mainly used in the literature~\cite{Bjorklund2007AGames,Fearnley2017EfficientGames,Jurdzinski2000AGames,Puri1995TheorySystems}.  
Based on these, many well-known improvement rules have exponential worst-case running time as demonstrated by sophisticated worst-case constructions, in particular by Friedmann et al. (see e.g. \cite{DBLP:journals/mp/DisserFH23, Friedmann2011ExponentialPrograms,Friedmann2013AMemorization}). 
The main idea behind the worst-case constructions is that one player can `trap' the other player repeatedly in different configurations so that the encountered strategies simulate a binary counter. 

While infinitary payoff games are symmetric in the two players by their nature, only some algorithms explicitly exploit this symmetry. 
Most attractor-based algorithms for parity games are inherently symmetric by simultaneously considering the game from the perspective of both players, while value iteration and strategy improvement algorithms are mostly inherently asymmetric.
Recently, a quasipolynomial symmetric algorithm for parity games was proposed by Jurdzinski et al.~\cite{JurdzinskiMorvanOhlmannThejaswini:2020}. 

Our work is mainly motivated by a symmetric version of strategy iteration for infinitary payoff games established in \cite{Schewe2015SymmetricImprovement}. 
This variant maintains two strategies, one for each player. 
The players then iteratively improve their strategy, using information from the best response to their opponent's strategy. 
This reduces the number of iterations needed in practice, and also does not have superpolynomial running time on the type of examples that were constructed for classic strategy improvement. 
The worst-case running time of this variant was unknown so far. 

\subparagraph{Our contribution} We develop a construction exhibiting exponential running time for symmetric strategy improvement (SSI). 
Our main result is the following:
\begin{theorem}\label{thm:worstcase}
In the worst case, the number of iterations of symmetric strategy improvement in parity games, mean payoff games and discounted payoff games is exponential in the number of nodes and edges in the graph \emph{independently} of the improvement rule. 
This holds for any of the currently used valuations in the literature.
\end{theorem}

It is remarkable that the result holds for any improvement rule.
This is different from regular strategy improvement, where the existence of a (theoretical) improvement rule for which the algorithm terminates in a linear number of iterations is guaranteed, see \cite[Lemma~4.2]{Friedmann2011ExponentialPrograms}.
By our Theorem~\ref{thm:worstcase}, this does not exist for symmetric strategy iteration. 

Moreover, we present a generalization of SSI which uses the valuation directly and not only the strategy of the opponent, allowing for more freedom to potentially overcome the exponential instances. 
However, we strengthen Theorem~\ref{thm:worstcase} with a subtle adapation of the worst-case instances to hold also for the generalization.  
This suggests that one needs a different approach involving more than only local information to benefit from insights in the interplay of strategies for both players. 

\subparagraph{Technical overview} To arrive at our main result, we derive a class of games for which SSI needs exponential running time. 
It is a careful adaptation (depicted in Figure~\ref{fig:sink_main}) of the basic example from \cite{Bjorklund2007AGames, gurvich1988cyclic} in such a way that the two players are both distracted by the other player's strategy. 
The key insight here is that the optimal counterstrategy to a bad strategy can also be a bad strategy itself. 
Hence restricting to moves from the optimal counterstrategy prevents them from making the crucial switches for achieving actual progress. 

Our family of games has a self-similar structure.
It requires the algorithm to solve a subgame first, and then after the important switches are made, solve the same subgame again, leading eventually to the exponential blowup of the number of iterations. 

Recall that symmetric strategy iteration picks only edges of the optimal counter strategy. 
While this implicitly also uses the valuation, as it is defined via the subgraph arising from the optimal counterstrategy, the generalization directly compares the valuations of the nodes. 
Only those edges are considered, which provide a local improvement over the valuations of both players.

To derive the lower bound construction for the generalization, we introduce a new gadget (Figure~\ref{fig:relsymmod}). 
This replaces each of the iteratively arranged pairs of nodes in the former family. 
The gadget then forces the generalization to exhibit a similar behaviour as the original SSI.

\subparagraph{Paper organization}
We provide the necessary background on parity games and (symmetric) strategy improvement in Section~\ref{sec:basics}. 
We introduce generalized symmetric strategy improvement in Section~\ref{sec:generalized-SSI}. 
Then, Section~\ref{sec:counterexample-basic} presents the structure and iterations of the basic exponential instance for SSI. 
This is refined in Sections~\ref{sec:adaptedexample} and~\ref{sec:concludingproof} to the generalized version of SSI and arbitrary improvement rules. 
We conclude with a discussion of potential extensions and limitations of our construction.

\section{Preliminaries}
\label{sec:basics}
\subparagraph{Parity games}
A parity game is a game played between two players, called player 0 and player 1. 
It is played on the nodes of a directed graph $G=(V,E)$, where the nodes are partitioned into $V=V_0\cup V_1$, and $V_i$ is controlled by player $i$. 
We assume every node has at least one outgoing edge. 
Associated with the nodes of the graph is a priority function $p \colon V\to Q$, where $Q\subset \mathbb{Z}$. 
At the start of the game, a pebble is placed on one of the nodes. 
A move is made by the player controlling the node that the pebble is currently on, and it consists of this player choosing an outgoing edge from this node. 
The pebble moves along the edge to the next node. 
The players keep making moves indefinitely. 
The winner of the game is decided by the largest node priority that the pebble encounters infinitely often. 
If it is even, player 0 wins, and if it is odd, player 1 wins.

It is well known that parity games are positionally determined, meaning that there always is a player that has a positional winning strategy. 
Positional means here that one only takes into account on which node the pebble currently is. 
Therefore, we define a strategy for player 0 as a function $\sigma \colon V_0\to V$ (with the condition that $(v,\sigma(v))\in E$ for all nodes $v \in V_0$). 
Similarly, a player 1 strategy is a function $\tau:V_1\to V$ (with $(v,\tau(v))\in E$ for all nodes $v \in V_1$).

\begin{figure}[t]
    \centering
    \includegraphics[width=0.8\linewidth]{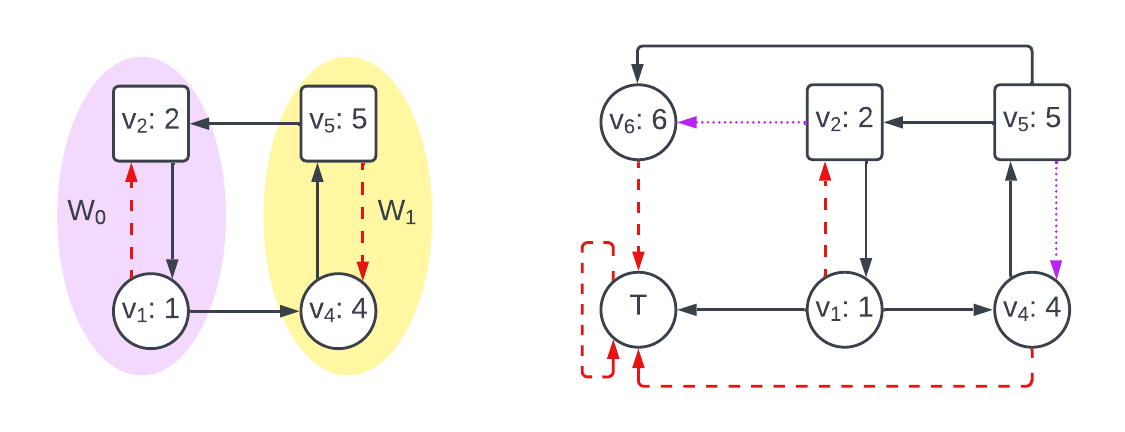}
    \caption{Left: A parity game. Priorities are shown in the nodes. Nodes controlled by player 0 are shown as circles, and nodes controlled by player 1 are squares. The sets of winning starting nodes for player 0 and 1 are $W_0$ and $W_1$, and the winning strategies are marked by dashed lines. Right: A sink parity game with a strategy $\sigma$ (dashed) and its optimal counterstrategy $\bar{\sigma}$ (dotted).}
    \label{fig:parityexample}
\end{figure}
\subparagraph{Sink parity games}
In this paper, we restrict ourselves to a class of parity games called \emph{sink parity games} as it allows to simplify the arguments for valuations.  
This class has often been used to show lower bounds for parity game algorithms (see e.g. \cite{Friedmann2011ExponentialPrograms, Hansen2012Worst-caseMethod}). 
Solving sink parity games is as hard as solving any parity game, see Lemma~\ref{lem:reducetosink}.
\begin{definition}
\label{def:sink}
A \emph{sink parity game} is a parity game that fulfills the following conditions:
\begin{enumerate}
    \item There exists a so-called \emph{sink node} $\top$ for which $p(\top)<p(v)$ for all nodes in $V\backslash \{\top\}$, and whose only outgoing edge is $(\top,\top)$.
    \item There exists a player 0 strategy $\sigma$ such that, when it is played, the highest priority (except~$\top$) in any possible cycle is even.
    \item There exists a player 1 strategy $\tau$ such that, when it is played, the highest priority (except~$\top$) in any possible cycle is odd.
\end{enumerate}
\end{definition}

We call player 0 and player 1 strategies that satisfy the above conditions \textit{admissible}. 
If player 0 plays an admissible strategy $\sigma$ and player 1 plays an admissible strategy $\tau$, then the pebble must end up at $\top$. 
Otherwise, it would enter the same node outside $\top$ twice, which would create a cycle, and the highest priority of the cycle would have to be odd and even at the same time. This also implies that $\top$ is reachable for the other player when $\sigma$ or $\tau$ is played.
One may say that the result of best play in a sink parity game is a `tie'.

\subparagraph{Valuations of strategies}
In the remainder of this section, we describe (symmetric) strategy improvement in parity games. We use the valuation of \cite{Fearnley2017EfficientGames} and \cite{Luttenberger2008StrategyGames}, adapted to sink parity games. The advantages of using this framework are that it is simpler to explain, and it focuses on the crucial second component of the most commonly used valuation by Jurdzinksi and V\"oge~\cite{Jurdzinski2000AGames}. In sink parity games and their related mean/discounted payoff games, the valuation that we describe here is equivalent to the other used versions of strategy improvement for parity games and mean payoff games \cite{Bjorklund2007AGames, Jurdzinski2000AGames, Puri1995TheorySystems}.

Now, suppose player 0 and player 1 both fix a strategy $\sigma$ and $\tau$, respectively. Then, given a starting node $v$, the course of the game is fixed. We want to evaluate how `good' this outcome is for player 0. This is expressed in the \textit{play value} $\Theta_{\sigma, \tau}(v)$. If the pebble does not reach $\top$, then it must eventually follow some cycle. 
If the highest priority in the cycle is even, then this is very good for player 0, so we assign $\infty$ to this play. If the priority is odd, we assign it value $-\infty$. Otherwise, the pebble follows a path to $\top$. 
In this case we establish the play value by counting how often each priority on this path is encountered. 
The even player aims to encounter many high even priorities and little high odd priorities on this path. 
The following definition formalizes this.

\begin{definition}
Let $\sigma$ and $\tau$ be a player 0 and player 1 strategy, respectively. Their \emph{play value} is a function $\Theta_{\sigma, \tau}:V\to \mathbb{Z}_{\geq 0}^Q\cup\{-\infty, \infty\}$, with $\mathbb{Z}_{\geq 0}^Q$ the set of nonnegative integer vectors indexed by the priority set $Q$. It is defined as follows:
\begin{itemize}
    \item Suppose the nodes encountered are $v=v_1,v_2,\ldots, v_k, \top, \top, \ldots$ (with $v_k\neq \top$). Then for any $q\in Q$, the $q$-element of $\Theta_{\sigma,\tau}$ (the component of the vector in $\mathbb{Z}_{\geq 0}^Q$ indexed by $q$) is given by
    $\left(\Theta_{\sigma,\tau}(v)\right)_q=|\{j\leq k: p(v_j)=q\}|$.
    \item If $\top$ is not reached and player 0 wins, then $\Theta_{\sigma,\tau}=\infty$.
    \item If $\top$ is not reached and player 1 wins, then $\Theta_{\sigma,\tau}=-\infty$.
\end{itemize} 
\end{definition}
We can compare play values by how `good' they are for player 0. 
This is done by a total order $\tle$ on $\mathbb{Z}_{\geq 0}^Q\cup\{-\infty, \infty\}$. 
The smallest element for $\tle$ is $-\infty$ and the largest is $\infty$. 
The order $\tle$ compares the play values in $\mathbb{Z}_{\geq 0}^Q$ lexicographically, but different for even and odd indices. To be precise, suppose we have $B$ and $C$ in $\mathbb{Z}_{\geq 0}^Q$, and that $B\neq C$. Let $q'$ be the highest priority $q$ for which $B_{q'}\neq C_{q'}$. We then say that $B\tl C$ if either $B_{q'} < C_{q'}$ and $q'$ is even, or $B_{q'} > C_{q'}$ and $q'$ is odd.
So if $B\tl C$, then $B$ has either less of some high even priority or more of some high odd priority, so $B$ is `worse' for player 0. 
Then we use the concept of play value to evaluate strategies. 
We evaluate a player 0 admissible strategy $\sigma$ by an optimal player 1 response $\bar{\sigma}$. 
To be precise, the valuation of an admissible strategy $\sigma$ is a function $\Xi_{\sigma}:V\to \mathbb{Z}_{\geq 0}^Q$ given by:
$\Xi_{\sigma}(v)=\min_{\tle}\{\Theta_{\sigma, \tau}(v)\:|\: \tau \text{ player 1 strategy}\}=\Theta_{\sigma, \bar{\sigma}}(v)$

A strategy $\tau$ attaining the minimum in the above equation is an \emph{(optimal) counterstrategy}.
Note that there exists a player 1 strategy $\bar{\sigma}$ that simultaneously attains the minimum for all nodes in $V$ (equation (11) in \cite{Jurdzinski2000AGames}). In general, this strategy might not be unique, but we pick one arbitrarily to be $\bar{\sigma}$ if there are multiple options.
Moreover, an optimal response $\bar{\sigma}$ can be computed efficiently. Because we have a sink parity game and an admissible $\sigma$, the minimum is never $\infty$ or $-\infty$, so we can regard the range of $\Xi_\sigma$ as $\mathbb{Z}_{\geq 0}^Q$.

The valuation of an admissible player 1 strategy $\tau$ is defined similarly using an optimal response $\bar{\tau}$ from player 0: $\Xi_{\tau}(v)=\max_{\tle}\{\Theta_{\sigma, \tau}(v)\:|\: \sigma \text{ player 0 strategy}\}=\Theta_{\bar{\tau}, \tau}$.

\smallskip

Figure~\ref{fig:parityexample} shows an example of a sink parity game with a strategy and its counterstrategy. 
For example, we have $\Xi_{\sigma}(v_2)=(0,1,0,0,0,1)$, as the play resulting from $\sigma$ and its optimal counterstrategy $\bar{\sigma}$ goes through the nodes with priorities 2 and 6. Likewise, \break $\Xi_{\sigma}(v_4)=(0,0,0,1,0,0)$, hence $\Xi_{\sigma}(v_2)\tr \Xi_{\sigma}(v_4)$.

\subparagraph{Strategy improvement}
The core idea behind strategy improvement is to make so-called improving moves. 
Improving moves for player 0 are given by the edges $(v,v')$ with $v\in V_0$ for which $\Xi_{\sigma}(v')\tr \Xi_{\sigma}(\sigma(v))$. 
That means that, with a new strategy picking $v'$ after $v$, player 0 can send the pebble to a node with higher valuation than it currently does in $\sigma$. 
We denote the set of improving moves for $\sigma$ by $I_{\sigma}$. 
Player 0 creates a new strategy $\sigma'$ from $\sigma$ by making improving moves, which means $\sigma'(v)=v'$ for a number of improving moves $(v,v') \in I_{\sigma}$ and $\sigma'(v)=\sigma(v)$ everywhere else. Of course, there may be multiple improving edges per node. The choice which edges to use to improve is decided by an \textit{improvement rule}, which is a function $f:\mathcal{P}(E)\to \mathcal{P}(E)$ that takes as input a set of improving edges, and outputs a set of improving edges subject to the following conditions:
\begin{itemize}
    \item If $|S|>0$, then $|f(S)|>0$.
    \item $f(S)\subseteq S$ for all $S\in \mathcal{P}(E)$. 
    \item Every node has at most one outgoing edge in $f(S)$.
\end{itemize}
By \cite{Jurdzinski2000AGames}, for any choice of improving edges, we have $\Xi_{\sigma'}(v)\tre \Xi_{\sigma}(v)$ for every $v\in V$, and $\Xi_{\sigma'}(v)\tr \Xi_{\sigma}(v)$ for at least one node $v$. Hence we increase the valuation of the strategy. Clearly, the new strategy is also admissible as its valuation is not $-\infty$. 
Additionally, if $\sigma$ has no improving moves, then we know that $\Xi_{\sigma}$ is pointwise maximal in the space of valuations (\cite[Lemma~5.8]{Jurdzinski2000AGames}). This leads to the strategy improvement algorithm (Algorithm~\ref{alg:stratit}).

\begin{algorithm}[t]
\caption{Strategy improvement}
\label{alg:stratit}
\begin{algorithmic}[1]
\State Start with some admissible strategy $\sigma$
\State Find an optimal counterstrategy $\bar{\sigma}$ to $\sigma$, compute $\Xi_{\sigma}$ and $I_{\sigma}$
\If{$I_{\sigma}=\emptyset$} \Return $\sigma$
\Else
\State Let $\sigma'$ be the strategy obtained from $\sigma$ by applying all improving moves from $f(I_{\sigma})$
\State $\sigma \gets \sigma'$, \textbf{go to} 2
\EndIf
\end{algorithmic}
\end{algorithm}

We can also define improving moves for player 1, by saying $(v,v')$ with $v\in V_1$ is an improving move if $\Xi_{\tau}(v')\tl \Xi_{\tau}(\tau(v))$. 
We denote the set of improving moves for player 1 by~$I_{\tau}$. 
Similar to before, if $\tau'$ is obtained from $\tau$ by making improving moves, $\Xi_{\tau'}(v)\tle \Xi_{\tau}(v)$ for every $v\in V$, and $\Xi_{\tau'}(v)\tl \Xi_{\tau}(v)$ for at least one node $v$. 
It is well-known that, if $\sigma$ is an optimal player 0 strategy (maximizing $\Xi_{\sigma}$ pointwise) and $\tau$ and an optimal player 1 strategy (minimizing $\Xi_{\tau}$ pointwise), then $\Xi_{\sigma}=\Xi_{\tau}$.

Now why are we interested in finding the strategy $\sigma$ that yields the highest valuation $\Xi_{\sigma}$ in a sink parity game? 
The winner of a sink parity game is already known, since the best both players can do is go to the sink node $\top$. 
However, even in a sink parity game, finding the optimal strategy (yielding the $\tle$-best $\Xi_{\sigma}$) is as difficult as solving parity games, as noted before in \cite{Friedmann2013AMemorization}. It is similar to the reduction to the longest shortest path problem in~\cite{Bjorklund2007AGames} and escape games in~\cite{DBLP:conf/csl/Schewe08}.

\begin{lemma}
\label{lem:reducetosink}
Deciding the winning starting sets and the winning strategies in a parity game can be polynomial-time reduced to finding a player 0 strategy $\sigma$ in a sink parity game that maximizes $\Xi_{\sigma}$.
\end{lemma}

\begin{claimproof}
Suppose we have a parity game $G=(V=V_0\cup V_1,E)$ with priority function $p:V\to \mathbb{Z}$. We may assume that there are no cycles within $V_0$ or within $V_1$. This is because we can always add nodes with small priorities controlled by player 0 in the middle of cycles in $V_1$ and vice versa without affecting the outcome of the game. Now we construct a parity game $G'=(V',E')$ from $G$, by adding two extra nodes, $\top$ and $w$. We extend $p$ by choosing $p(\top)$ smaller than all other priorities, and taking for $p(w)$ an even number higher than all other priorities. We add an edge from every node in $V_0$ to $\top$, from every node in $V_1$ to $w$, from $w$ to $\top$ and from $\top$ to $\top$.

\begin{figure}[t]
    \centering
    \includegraphics[width=0.6\linewidth]{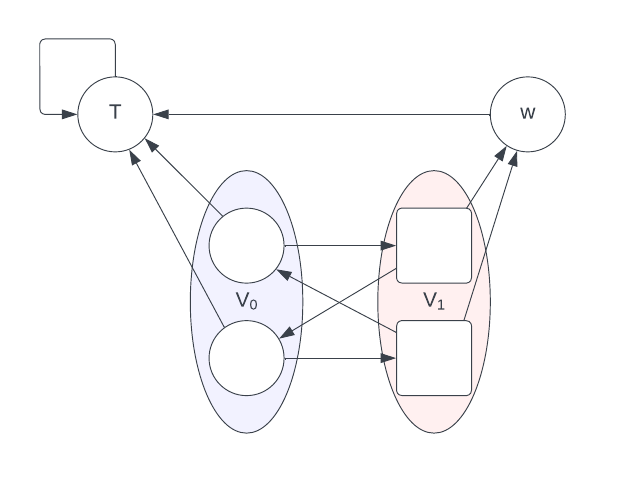}
    \caption{Reduction to a sink parity game}
    \label{fig:sinkgamereduction}
\end{figure}
This is clearly a sink parity game, since player 0 has an admissible strategy by always going to $\top$, and player 1 has an admissible strategy by going to $w$. Now let $\sigma$ be an optimal player 0 strategy that maximizes $\Xi_{\sigma}$ pointwise, and let $\bar{\sigma}$ be player 1's optimal response. Since $p(w)$ is even and very large, if player 1 can avoid entering $\top$ through $w$, they will do so. Define the subgraph $G'_{\sigma}$ by the graph with the same node set as $G'$ and with edge set $\{(v,\sigma(v)): v\in V_0'\}\cup \{(v,v')\in E': v\in V_1'\}$. Suppose for a node $v$ that $\left(\Xi_{\sigma}(v)\right)_{p(w)}=1$. This implies that in $G'_{\sigma}$, the node $\top$ is only reachable from $v$ through $w$. In particular, this means that player 1 would always have to end in a cycle in $V$ if they would not have the option of going to $w$. Because $\sigma$ is admissible, this cycle has an even highest priority. This implies that, in the original game $G$, player 0 wins the game that starts at $v$ by playing $\sigma$. If, on the other hand, $\left(\Xi_{\sigma}(v)\right)_{p(w)}=0$, this implies that $\left(\Xi_{\tau}(v)\right)_{p(w)}=0$, and we can argue in the same way that player 0 can only reach cycles of odd priority in $V$ if player 1 plays $\tau$. Hence $\tau$ wins for player 1 in the parity game $G$ that starts from $v$. So we found the winning starting sets and winning strategies for both players in $G$. (note that we could also have made the above construction with $p(w)$ odd and connecting player 0 nodes to $w$ and player 1 nodes with $\top$). See also Example \ref{ex:reduce}. 
\end{claimproof}

\subparagraph{Symmetric strategy improvement} \label{symmetricalgorithm}
The symmetric strategy improvement (SSI) algorithm was introduced by Schewe et al. in~\cite{Schewe2015SymmetricImprovement} as a symmetric version of strategy improvement. 
The algorithm maintains and improves two strategies simultaneously: a player 0 strategy $\sigma$ and a player 1 strategy $\tau$. 
It uses an optimal counterstrategy $\bar{\tau}$ to $\tau$ to select improving moves for $\sigma$, and an optimal counterstrategy $\bar{\sigma}$ to $\sigma$ to select improvements for $\tau$. 
Note that this could be applied to a broader class of games than just (sink) parity games, in particular mean and discounted payoff games. 
It is described in Algorithm~\ref{alg:symstratit} where we expicitly include the choice of an improvement rule  $f \colon \mathcal{P}(E) \to \mathcal{P}(E)$. 

\begin{algorithm}[t]
\caption{Symmetric strategy improvement}
\label{alg:symstratit}
\begin{algorithmic}[1]
\State Start with some pair of admissible strategies $\sigma$,$\tau$
\State Find counterstrategies $\bar{\sigma}$ and $\bar{\tau}$ and compute $I_{\sigma}$ and $I_{\tau}$
\State $I\gets f\left((I_{\sigma}\cap \{(v,\bar{\tau}(v))| v\in V_0\})\cup (I_{\tau}\cap \{(v,\bar{\sigma}(v))| v\in V_1\}))\right)$
\State Let $\sigma',\tau'$ be result of applying all improving moves from $I$ to $\sigma, \tau$
\If{$\sigma=\sigma'$ \textbf{ and } $\tau=\tau'$} \textbf{return }$\sigma, \tau$
\Else
\State $\sigma\gets \sigma'$, $\tau \gets\tau'$, \textbf{go to} 2
\EndIf
\end{algorithmic}
\end{algorithm}

It is clear that the algorithm terminates, since any improving move improves the valuation for the respective player, and there is only a finite number of strategies (having a fixed valuation) for both players. 
The following lemma implies that the algorithm only terminates when the resulting pair of strategies $(\sigma,\tau)$ is optimal for the players. 
It is implicitly proven in \cite[Lemma~3]{Schewe2015SymmetricImprovement}.

\begin{lemma} \label{lem:symworks}
Suppose $\sigma$ is a non-optimal player 0 strategy or $\tau$ is a non-optimal player 1 strategy. 
Let $\bar{\sigma}$ and $\bar{\tau}$ be optimal counterstrategies to $\sigma$ and $\tau$, respectively. Then at least one of the sets $I_{\sigma}\cap \{(v,\bar{\tau}(v))| v\in V_0\}$ and $I_{\tau}\cap \{(v,\bar{\sigma}(v))| v\in V_1\}$ is nonempty.
\end{lemma}

\section{Generalized symmetric strategy improvement}
\label{sec:generalized-SSI}

It was an intriguing insight by Schewe et al. how the interplay between strategies of both players can be used to overcome the known hard instances of strategy improvement. 
Since the evaluation of the goodness of a strategy relies on the valuation, we go one step further and directly use the interplay between the valuations arising from the strategies of the two players to make the improvements. 
While this can overcome the first basic family of hard instances presented in Section~\ref{sec:counterexample-basic}, we show in Section~\ref{sec:adaptedexample} that actually both versions can still be forced to take exponentially many steps. 
In this way, we extend the result of Theorem~\ref{thm:worstcase}.

Our selection of improving edges contains the original set but is bigger in general. Instead of basing the choices just on optimal counterstrategies, we construct sets $J_{\sigma}(\tau)$ and $J_{\tau}(\sigma)$ by directly comparing locally the valuations of adjacent nodes: 
\begin{align*}
J_{\sigma}(\tau) &=\{(v,w):\; v\in V_0\wedge\Xi_{\tau}(w)\tre \Xi_{\tau}(\sigma(v))\}, \\
J_{\tau}(\sigma) &=\{(v,w):\; v\in V_1\wedge\Xi_{\sigma}(w)\tle \Xi_{\sigma}(\tau(v))\}.
\end{align*}
These notions allow us to state our generalized version in Algorithm \ref{gssi}. 
They are in some sense the counterparts of $I_{\sigma}$ and $I_{\tau}$, recall that $I_{\sigma}=\{(v,w):\; v\in V_0\wedge\Xi_{\sigma}(w)\tr \Xi_{\sigma}(\sigma(v))\}$. 
One obtains the original SSI back in this context by choosing only those improvement rules that select only edges used by $\bar{\sigma}$ and $\bar{\tau}$.  

\begin{algorithm}[t]
\caption{Generalized symmetric strategy improvement}
\label{gssi}
\begin{algorithmic}[1]
\State Start with some admissible strategies $\sigma$ and $\tau$
\State Find $I_{\sigma}$, $I_{\tau}$, $J_{\sigma}(\tau)$ and $J_{\tau}(\sigma)$ 
\State $I\gets f\left(\left(I_{\sigma}\cap J_{\sigma}(\tau)\right)\cup\left(I_{\tau}\cap J_{\tau}(\sigma)\right)\right)$
\State Let $\sigma',\tau'$ be result of applying all improving moves from $I$ to $\sigma, \tau$
\If{$\sigma=\sigma'$ \textbf{ and } $\tau=\tau'$} \textbf{return }$\sigma, \tau$
\Else  
\State $\sigma\gets\sigma'$, $\tau\gets\tau'$, \textbf{go to }2
\EndIf
\end{algorithmic}
\end{algorithm}

\subparagraph{Correctness of generalized symmetric strategy improvement}\label{correct}
Like for normal SSI, it is clear that the algorithm terminates. 
We are left to show that there is always an improving move possible if the pair of strategies is not optimal. 
We do so by showing that all the improving moves that were possible in SSI are also possible in the generalization. 
Correctness of Algorithm~\ref{gssi} then follows from Lemma~\ref{lem:symworks}. 
\begin{lemma}
For any pair of strategies $\sigma,\tau$, we have that any edge $(v,\bar{\sigma}(v))$ is in $J_{\tau}(\sigma)$, and any edge $(v,\bar{\tau}(v))$ is in $J_{\sigma}(\tau)$.  
\end{lemma}
\begin{claimproof}
We do so by contradiction. 
Suppose there is an edge $(v, \bar{\sigma}(v))$ that is not in $J_{\tau}(\sigma)$. This means by definition of $J_{\tau}(\sigma)$ that $\Xi_{\sigma}(\bar{\sigma}(v))\tr \Xi_{\sigma}(\tau(v))$. 
Note that the valuation of $\sigma$, which is $\Xi_{\sigma}$, is equal to $\Theta_{\sigma\bar{\sigma}}$. Now consider the strategy subgraph $G_{\sigma}:=(V,E_{\sigma})$ with $E_{\sigma}:=\{(v,w):v\in V_1\}\cup\{(v, \sigma(v)):v\in V_0\})$. In $G_{\sigma}$, the valuation of $\bar{\sigma}$ is equal to $\Theta_{\sigma\bar{\sigma}}=\Xi_{\sigma}$, as there is only one player 0 strategy possible. But then $\Xi_{\sigma}(\bar{\sigma}(v))\tr \Xi_{\sigma}(\tau(v))$ implies that $(v,\tau(v))$ is an improving move for player 1 strategy $\bar{\sigma}$ in the game $G_{\sigma}$. 
However, this is not possible, since $\bar{\sigma}$ is defined to be an optimal counterstrategy to $\sigma$. We conclude that $(v, \bar{\sigma}(v))\in J_{\tau}(\sigma)$. The proof of $(v, \bar{\tau}(v))\in J_{\sigma}(\tau)$ is analogous.
\end{claimproof}

\section{Counterexample for worst-case complexity}
\label{sec:counterexample-basic}

We present a family of parity games for which the original SSI needs exponentially many iterations, where we restrict to the SWITCH-ALL improvement rule at first. 
This rule selects one improving edge for every node that has an outgoing improving edge in $S$. 
We generalize this to arbitrary improvement rules in Section~\ref{sec:concludingproof}. Some animations of the iterations of the examples presented in Sections \ref{sec:counterexample-basic} and \ref{sec:adaptedexample} can be found at \url{https://github.com/MatthewMaat/SI-animations/tree/master/Example%20animations}

\smallskip

We use a sequence of sink parity games $(G_n)_{n\in \mathbb{N}}$ whose nodes and edges of $G_n$ are listed in Table~\ref{tab:ladder} and Figure~\ref{fig:sink_main}. 
The main structure is similar to the mean payoff game from \cite{Bjorklund2007AGames} and \cite{gurvich1988cyclic}, with the main difference being that we have backward edges to the `start' of the game.

\begin{table}[b]
    \centering
    \begin{tabular}{c|c|c|c}
       Node  &  Player  & Priority  & Successors \\ \hline
       $a_1$  &  Player 0 & 3 &  $a_2$, $d_2$\\
       $a_i$, $i=2,\ldots, n$  &  Player 0 & $2i+1$ &  $a_1$, $a_{i+1}$, $d_{i+1}$\\ 
       $d_1$  &  Player 1 & 4 &  $a_2$, $d_2$\\
       $d_i$, $i=2,\ldots, n$  &  Player 1 & $2i+2$ &  $d_1$, $a_{i+1}$, $d_{i+1}$\\ 
       $a_{n+1}$ & Player 0 & $1$ & $a_{n+1}$\\
       $d_{n+1}$ & Player 1 & $2n+4$ & $a_{n+1}$
    \end{tabular}
    \caption{Nodes and edges of $G_n$}
    \label{tab:ladder}
\end{table}

\begin{figure}[t]
    \centering
    \includegraphics[width=\linewidth]{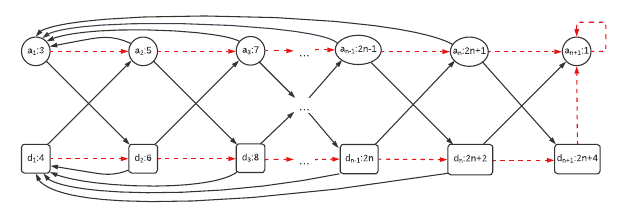}
    \caption{The graph $G_n$, with the priorities written in the nodes. The initial strategies $\sigma_0$ and $\tau_0$ are dashed.}
    \label{fig:sink_main}
\end{figure}

\smallskip

\textit{Initial strategies} We define $\sigma_0$ and $\tau_0$ by $\sigma_0(a_i)=a_{i+1}$ and $\tau_0(d_i)=d_{i+1}$ for $i=1,2,\ldots,n$, $\sigma_0(a_{n+1})=a_{n+1}$ and $\tau_0(d_{n+1})=a_{n+1}$. 

\textit{$G_n$ is a sink parity game.}  The node $a_{n+1}$ is the sink node with low priority. Also, $\sigma_0$ is an admissible player 0 strategy. If the pebble enters any node $a_i$ if $\sigma_0$ is played, it will end at the sink. The only way player 1 could try to avoid this is by trying to keep the pebble within $d_1, d_2, \ldots, d_n$, but then this will create a cycle with even highest priority. Hence $\sigma_0$ is admissible. Likewise, $\tau_0$ is admissible, so we have a sink parity game.

\textit{Optimal strategies of the players.} 
Recall that a strategy having a high valuation (for player 0) corresponds to being able to pass through node with high even priorities and avoiding nodes with high odd priorities. 
Therefore, to maximize their valuation, player 0 lets the pebble pass $d_{n+1}$, which has the largest (and even) priority in the game. 
Playing a strategy $\sigma$ that achieves this when starting from a node $v$ yields $\left(\Xi_{\sigma}(v)\right)_{2n+4}=1$. 
This means the valuation is larger than any strategy that does not achieve this. 
There is only one strategy where player 0 can do this from every starting node $a_i$, namely the strategy with $\sigma(a_i)=a_{i+1}$ for $i<n$ and $\sigma_{a_n}=d_{n+1}$. 
Ironically, this differs from $\sigma_0$ in only one edge. 
However, as we will see, we avoid making this switch for a long time. Likewise, player 1's goal is to avoid $d_{n+1}$. 
Their only strategy to avoid it from every node $d_i$ is to pick $\tau(d_i)=d_{i+1}$ for $i<n$, and $\tau(d_n)=a_{n+1}$. 
Again, player 1 is only one switch from optimal when starting with $\tau_0$. 

\bigskip

The remainder of this section is dedicated to the proof of the following proposition, where we consider iterations in $5$ phases. 
In Section~\ref{sec:concludingproof}, we discuss how our main result follows from this.

\begin{proposition}\label{lem:induction}
Suppose SSI with the SWITCH-ALL rule on the game graph $G_n$ starts with the strategy pair $\sigma_0, \tau_0$. 
Then after $2^{n+1}-3$ iterations, the optimal strategies $\sigma$ and $\tau$ are found. 
The optimal strategies for both players do not appear in any earlier iteration.
\end{proposition}

\begin{proof} 
The reader can verify that one needs only one iteration to reach the optimum in $G_1$. 
We assume the claim to be true for $G_{j-1}$ and show that it holds for $G_j$ to conclude the proof by induction. We do this by showing that the iterations of SSI are as follows:
\begin{enumerate}
    \item $2^j-3$ iterations with switches from $a_1, a_2, \ldots, a_{j-1}$ and $d_1,d_2, \ldots, d_{j-1}$
    \item One iteration where the only switch made is $(a_j,a_1)$
    \item One iteration where the only switch made is  $(d_j,a_{j+1})$
    \item One iteration where the only switch made is  $(a_j,d_{j+1})$
    \item $2^j-3$ iterations with switches from $a_1, a_2, \ldots, a_{j-1}$ and $d_1,d_2, \ldots, d_{j-1}$
\end{enumerate}
Then clearly the total number of iterations is $2\cdot (2^j-3)+3=2^{j+1}-3$.
Now, we elaborate on these $5$ steps, where the first $2^j-3$ iterations need extra insights captured in the following three observations.

\begin{observation}
    \label{obs:counter}As long as $\sigma(a_j)=a_{j+1}$, we have $\bar{\sigma}(d_j)=d_1$. Likewise, as long as $\tau(d_j)=d_{j+1}$, we have $\bar{\tau}(a_j)=a_1$
\end{observation}
\begin{claimproof} Suppose $\sigma(a_j)=a_{j+1}$, and we look at what $\bar{\sigma}$ could be. Recall that player 1's goal is to reach $a_{j+1}$ without passing $d_{j+1}$. Starting from $d_{j}$, player 1 can do so in two ways: by picking $\bar{\sigma}(d_j)=a_{j+1}$, or by picking $\bar{\sigma}(d_j)=d_1$ and then choosing $\bar{\sigma}$ for $d_1, d_2, \ldots, d_{j-1}$ such that the pebble ends up at $a_j$. The first option gives $\left(\Xi_{\sigma}(d_j)\right)_{2j+1}=0$ while the second one gives $\left(\Xi_{\sigma}(d_j)\right)_{2j+1}=1$. Hence, the latter gives a lower valuation, and $\bar{\sigma}(d_j)=d_1$. We can prove similarly that while $\tau(d_j)=d_{j+1}$, we have $\bar{\tau}(a_j)=a_1$.
\end{claimproof}

\begin{observation}
    \label{obs:value}Suppose $\sigma(a_j)=a_{j+1}$ and $\tau(d_j)=d_{j+1}$. Then $\Xi_{\sigma}(a_j)\tl \Xi_{\sigma}(d_j)$ and $\Xi_{\tau}(a_j)\tl \Xi_{\tau}(d_j)$. In both cases, the valuation vectors compared differ in their $q$-position where $q\geq 2j+1$.
\end{observation}
\begin{claimproof}From the proof of Observation \ref{obs:counter} we know that when the pebble starts at $d_j$ and players play $\sigma, \bar{\sigma}$, the pebble goes to $d_1$, then to $a_j$ and directly to $a_{j+1}$. Hence $\Xi_{\sigma}(a_j)$ and $\Xi_{\sigma}(d_j)$ differ in their $p(d_j)=2j+2$-component, where the latter valuation has a 1. Then it follows that $\Xi_{\sigma}(a_j)\tl \Xi_{\sigma}(d_j)$, as the values in the vector corresponding to smaller priorities are insignificant when comparing these two valuations. We can likewise see that $\Xi_{\tau}(a_j)\tl \Xi_{\tau}(d_j)$ because they differ in their $p(a_j)=2j+1$-component.
\end{claimproof}

\begin{observation}\label{obs:improving}Suppose $\sigma(a_j)=a_{j+1}$ and $\tau(d_j)=d_{j+1}$. Then edge $(a_j,a_1)$ is improving only if $\sigma(a_i)=a_{i+1}$ for $i<j-1$ and $\sigma(a_{j-1})=d_j$. Edge $(d_j,d_1)$ is never improving.
\end{observation}
\begin{claimproof} Suppose player 0 switches the improving edge $(a_j,a_1)$ while $\sigma$ is different than specified above. The resulting strategy should also be admissible. However, player 1 can play strategy $\bar{\sigma}(d_i)=d_{i+1}$ for $i<j-1$ and $\bar{\sigma}(d_{j-1})=a_j$, creating a cycle with highest priority $p(a_j)=2j+1$. So the new strategy is not admissible, so $(a_j,a_1)$ could not have been improving. Likewise, suppose $(d_j,d_1)$ is improving for some strategy $\tau$ and we switch it. Then player 0 can play $\bar{\tau}(a_i)=a_{i+1}$ for $i<j-1$ and $\bar{\tau}(a_{j-1})=d_j$. This either creates a cycle with highest priority $p(d_j)=2n+2$ or player 1 creates another cycle within $d_1, d_2, \ldots, d_{j-1}$. In both cases, the resulting player 1 strategy cannot be admissible, so the edge could not have been improving.
\end{claimproof}

From Observations \ref{obs:counter} and \ref{obs:improving}, we can conclude that the strategies at $a_j$ and $d_j$ will not change until $\sigma$ fulfills the conditon of Observation \ref{obs:improving} ($\sigma(a_i)=a_{i+1}$ for $i<j-1$ and $\sigma(a_{j-1})=d_j$). Moreover, from Observation \ref{obs:value}, we notice that at this first part of the algorithm we may as well remove $a_{j+1}$ and $d_{j+1}$, and set $p(a_j)=1$ and add edges $(d_j,a_j)$, $(a_j,a_j)$ without changing any choices of SSI. But now we are left with an exact copy of $G_{j-1}$, so by induction hypothesis, we know that SSI needs $2^j-3$ iterations to reach the optimal pair of strategies there. The optimal player 0 strategy in $G_{j-1}$ is $\sigma(a_i)=a_{i+1}$ for $i<j-1$ and $\sigma(a_{j-1})=d_j$, which we only reach after $2^j-3$ iterations. So only then do we need to consider the original graph $G_j$ again and are we able to switch edge $(a_j,a_1)$ by Observation \ref{obs:improving}. 

\begin{figure}[t]
    \centering
    \includegraphics[width=\linewidth]{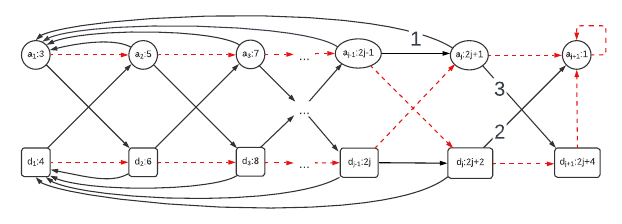}
    \caption{The strategies $\sigma,\tau$ (dashed) after $2^j-3$ iterations of SSI. Edges that are switched in the next 3 iterations are marked with 1,2,3.}
    \label{fig:my_label}
\end{figure}
\smallskip

\textit{The $(2^j-2$)-th iteration} There are no improving moves in $a_1, \ldots, a_{j-1}$ or in $d_1,\ldots, d_{j-1}$, as the strategies are 'optimal' strategies in $G_{j-1}$. However, we do not have optimal strategies in $G_j$ yet, so by Lemma \ref{lem:symworks}, there must be at least one edge switched by SSI. The only choice left is edge $(a_j, a_1)$, so this edge is now switched. 

\smallskip

\textit{The $(2^j-1)$-th iteration} 
We have $\sigma=\bar{\tau}$ (recall $\bar{\tau}$ from Observation~\ref{obs:counter}). 
Since $\tau$ is the same as in the last iteration, the nodes $d_1, \ldots, d_{j-1}$ do not have improving moves. 
Only the edge $(d_j, a_{j+1})$ can be switched. 
Since the strategies are not optimal yet, again Lemma~\ref{lem:symworks} implies that we switch exactly this edge in the $(2^j-1)$-th iteration.

\smallskip

\textit{The $2^j$-th iteration} 
At the start of the $2^j$-th iteration, we observe that player 1's strategy $\tau$ is equal to the counterstrategy $\bar{\sigma}$. 
This is since player 1 can only reach the sink (avoiding $d_{j+1}$) by playing $\bar{\sigma}(d_j)=a_{j+1}$, and  only with this current strategy can player 1 additionally pass node $a_{j}$ on the way there if player 0 plays $\sigma$. 
Hence, SSI does not make any improving moves for player 1. 
If player 0 makes any switch in the nodes $a_1, a_2, \ldots, a_{j-1}$, then this always allows player 1 to create a cycle with odd highest priority, so this cannot be an improving move. 
Therefore, SSI can only make a switch in $a_j$. 
The only improving move in $a_j$ is $(a_j,d_{j+1})$, so exactly this edge is switched in iteration $2^j$.

\smallskip

\textit{The final $2^j - 3$ iterations}  
Notice that we will never again make switches in nodes $a_j$ and $d_j$. 
Moreover, $\Xi_{\sigma}(a_j)\tr \Xi_{\sigma}(d_j)$ and $\Xi_{\tau}(a_j)\tr \Xi_{\tau}(d_j)$. 
We might as well replace $d_j$ by a sink node of priority 1, and $a_j$ by a node with priority $2j+2$ with an edge to the sink (and remove $a_{j+1},d_{j+1}$). 
This is without changing any future iterations of SSI. 
But then we have again a copy of $G_{j-1}$ with its respective starting strategies. 
Using the induction hypothesis again, SSI takes another $2^j-3$ iterations. 
It reaches the optimal strategies only in the last iteration. 
These strategies are also optimal in $G_j$, and this completes the proof of Proposition~\ref{lem:induction}.
\end{proof} 

\section{Adapted counterexample for generalization of symmetric strategy improvement}
\label{sec:adaptedexample}

We consider the worst-case performance of the generalized version of SSI, when the SWITCH-ALL rule is used. 
The generalized version can solve the games from the previous section quickly, as it considers more improving moves.
However, in this section we show that the generalized version still has exponential running time on a suitably modified version of the counterexample.

The overall structure is the same as for the original counterexample, but the nodes $a_i$ and $b_i$ are replaced by gadgets as shown in Figure~\ref{fig:relsymmod}. The full construction is described in Table \ref{tab:module}.
The nodes $a_i$ and $d_i$ (except the sink) have priorities larger than $N$ and the other nodes have priority smaller than $N$, so the priorities of $a_i$ and $d_i$ are still the most important ones. 

The function of these gadgets is to make it harder for the players to switch. 
Now for example, instead of just making a switch at $a_i$, player 0 has to switch their choice at both $c_i$ and $m_i$ to significantly change the course of the game. Additionally, the nodes $e_i,f_i,k_i,l_i$ make improving moves seem very insignificant with respect to differences in valuation. Similar to Proposition \ref{lem:induction}, the following proposition holds: 

\begin{proposition}
    \label{lem:induction2}
    Suppose generalized SSI on the game graph $G_n$ starts with $\sigma_0, \tau_0$. Then after $7\cdot 2^{n-1}-5$ iterations, the optimal strategies $\sigma$ and $\tau$ are found. 
    The optimal strategies for both players do not appear in any earlier iteration.
\end{proposition}
\begin{proof}
\begin{figure}[t]
    \centering
    \includegraphics[width=\linewidth]{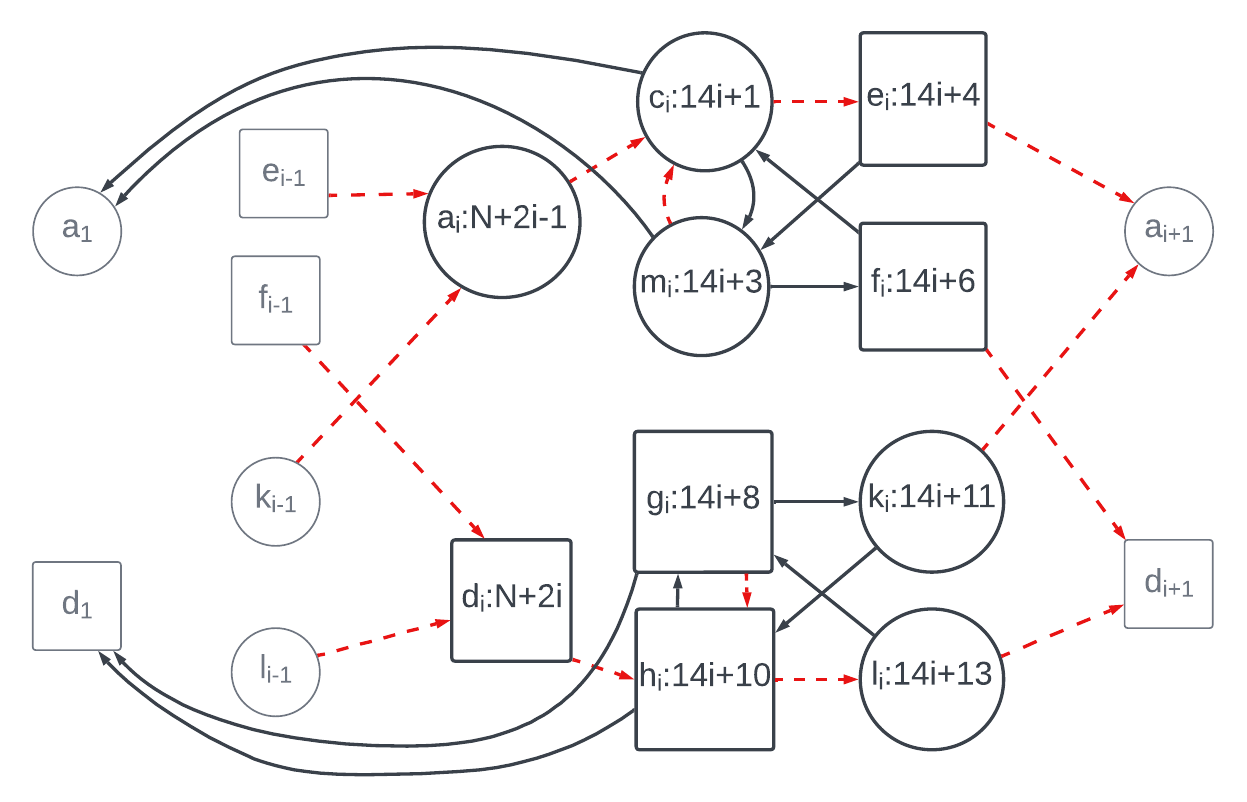}
    \caption{Subgraph for generalized SSI that replaces $a_i$ and $b_i$ when $i>1$. Initial strategies $\sigma_0$ and $\tau_0$ are dashed. Nodes that do not have all their incident edges shown are small and grey.}
    \label{fig:relsymmod}
\end{figure}
\begin{table}[b]
    \centering
    \begin{tabular}{c|c|c|ccc|c|c|c}
       Node  &  Player  & Priority  & Successors & & Node  &  Player  & Priority  & Successors \\ \cline{1-4} \cline{6-9}
       $a_i$ &  Player 0 & $N+2i-1$ &  $c_i$ & & $a_{n+1}$ & Player 0 & $1$ & $a_{n+1}$\\ 
       $d_i$  &  Player 1 & $N+2i$ &  $h_i$ & & $d_{n+1}$ & Player 1 & $N+2n+2$ & $a_{n+1}$\\ 
       $c_i$ & Player 0 & $14i+1$ & $e_i,m_i$, $(a_1)$ & & $e_i$ & Player 1 & $14i+4$ & $m_i,a_{i+1}$\\
       $m_i$ & Player 0 & $14i+3$ & $f_i,c_i$, $(a_1)$ & & $f_i$ & Player 1 & $14i+6$ & $c_i,d_{i+1}$\\
       $g_i$ & Player 1 & $14i+8$ & $k_i,h_i$, $(d_1)$ & & $k_i$ & Player 0 & $14i+11$ & $h_i,a_{i+1}$\\
       $h_i$ & Player 1 & $14i+10$ & $l_i,g_i$, $(d_1)$ & & $l_i$ & Player 0 & $14i+13$ & $g_i,d_{i+1}$\\     
    \end{tabular}
    \begin{tabular}{cccc}
         $\sigma_0(a_i)=c_i$ & $\tau_0(d_i)=h_i$ & $\sigma_0(a_{n+1})=a_{n+1}$ & $\tau_0(d_{n+1})=a_{n+1}$\\
         $\sigma_0(c_i)=e_i$ & $\sigma_0(m_i)=c_i$ & $\tau_0(g_i)=h_i$ & $\tau_0(h_i)=l_i$\\
         $\tau_0(e_i)=a_{i+1}$ & $\tau_0(f_i)=d_{i+1}$ & $\sigma_0(k_i)=a_{i+1}$ & $\sigma_0(l_i)=d_{i+1}$
    \end{tabular}
    \caption{Nodes, edges and initial strategies of the adapted counterexample. We have $N=16n+16$, and $i$ ranges from 1 to $n$. Nodes between brackets mean that they are only a successor if $i>1$.}
     \label{tab:module}
\end{table}
We show this by induction similar to before. For $n=1$, the reader can verify that in the first iteration we switch $(m_1,f_1)$ and $(g_1,k_1)$, and in the second iteration we switch $(h_1,g_1)$ and $(c_1,m_1)$ to reach the optimum. For the induction step, we assume the lemma holds for $G_{j-1}$, and show that the iterations of generalized SSI on $G_j$ are as follows:
\begin{enumerate}
    \item $7\cdot2^{j-1}-5$ iterations with switches in nodes $x_i$ with index $i<j$
    \item One iteration where only $(c_j,a_1)$ and $(m_j,a_1)$ are switched
    \item One iteration where only $(g_j,k_j)$ is switched
    \item One iteration where only $(h_j,g_j)$ is switched
    \item One iteration where only $(m_j,f_j)$ is switched
    \item One iteration where only $(c_j,m_j)$ is switched
    \item $7\cdot2^{j-1}-5$ iterations with switches in nodes $x_i$ with index $i<j$
\end{enumerate}

Similar to the proof of Proposition \ref{lem:induction}, we can argue about the phase of the algorithm where no switches are yet made in the nodes $c_j, m_j, g_j$ and $l_j$. We can observe that edges $(c_j,a_1)$ and $(m_j,a_1)$ are part of $\bar{\tau}$ but not yet improving until player 0 plays a specific strategy, and $(g_j,d_1)$ and $(h_j,d_1)$ are part of $\bar{\sigma}$ and never improving. So in the first phase, the only improving moves with index $j$ are $(m_j,f_j)$ and $(g_j,k_j)$. But we do not make these moves for a reason that is similar to Observation \ref{obs:counter}. If we consider the path of the pebble when $\tau$ and $\bar{\tau}$ are played, then starting from $f_j$, we go immediately to $d_{j+1}$, while from $c_j$ we go back to $a_1$ and through $d_j$ to $d_{j+1}$. So $\Xi_{\tau}(c_j)$ is much bigger (because of the $p(d_j)=N+2i$) than $\Xi_{\tau}(f_j)$. Hence generalized SSI never makes the switch $(m_j,f_j)$ in the beginning. Likewise, the move $(g_j,k_j)$ is postponed by the algorithm. Therefore, we can again pretend that $a_j$ is the sink and $d_{j}$ only has an edge towards the sink and use the induction hypothesis to show that this first phase takes $7\cdot 2^{j-1}-5$ iterations.

Now going back to the induction proof, the next five iterations of the algorithm are similar to the three iterations in the middle of the counterexample for SSI, except that we need two switches per module instead of one.

For the last part, we can again pretend that $d_j$ is the sink and $a_j$ a node with even priority with an edge to the sink. This yields a copy of $G_{j-1}$, but with nodes like $c_{j-1}$ and $e_{j-1}$ and $(f_{j-1})$ switched. This, however, does not affect any choices of the algorithm (as the valuations of $c_{i}$ and $e_i$ are always very close). Then using the induction hypothesis, it follows that this last phase takes $7\cdot 2^{j-1}-5$ iterations. This completes the induction proof.
\end{proof}

\section{Concluding the proof}
\label{sec:concludingproof}
We look at the last details to prove Theorem \ref{thm:worstcase}.
So far, we assumed that the improvement rule SWITCH-ALL is used. 
It turns out that in many iterations of (generalized) SSI, there is only one improving move. 
This implies the following result.

\begin{lemma}
    The results of Propositions \ref{lem:induction} and \ref{lem:induction2} hold for any improvement rule $f$.
\end{lemma}

\begin{claimproof}
    We prove this by arguing that if we use SWITCH-ALL on the graphs $G_n$ for both constructions, then either there is only one switch possible, or every switch but one switch is irrelevant. Here irrelevant means that whether or not we make the switch, the further course of the algorithm does not change, hence the number of iterations can only increase if we decide not to make some improving switches.
    
    First, we consider the class of games $G_n$ that proved Proposition \ref{lem:induction} (Figure  \ref{fig:sink_main}). If we use the SWITCH-ALL rule, then by following the induction proof we find that we make two switches per iteration whenever we make switches at $a_1$ and $d_1$ (following from the induction basis), and we make one switch in every other case (from the induction step). Note however, that not making a switch in $d_1$ never has any effect on the course of the algorithm. Observations \ref{obs:counter}, \ref{obs:value}, \ref{obs:improving} still hold if no switch is made in $d_1$, and so does the rest of the induction step. It also follows that if only $d_1$ is switched and $a_1$ not, then in the next iteration $a_1$ has the only possible switch. In all cases, symmetric strategy improvement takes at least as many iterations as for SWITCH-ALL.

    Next we consider the games $G_n$ from the proof of Proposition \ref{lem:induction2} (Figure \ref{fig:relsymmod}). There are two cases where there are multiple switches when SWITCH-ALL is used. First, the modules attached to $a_1$ and $d_1$ are switched at the same time. However, for the same reason as above, we can ignore the second module. Secondly, $(c_i,a_1)$ and $(m_i,a_1)$ are switched at the same time. Note that if $(c_i, a_1)$ is switched, then any switch in $m_i$ becomes irrelevant as $m_i$ is not reachable at that moment for player 0, and we make the switch $(m_i,f_i)$ three iterations later. If, on the other hand, only $(m_i,a_1)$ is switched while $\sigma(c_i)=e_i$, then nothing significant changes in the game and in the next iteration we will still have to switch $(c_i,a_1)$. Finally, if $(m_i,a_1)$ is switched while $\sigma(c_i)=m_i$ and $\sigma(m_i)=f_i$, then the switch $(c_i,a_1)$ becomes irrelevant after the switch as it only removes one $14i+3$-priority from the valuation. So in all cases, there is only one relevant switch. We conclude that using an improvement rule for this $G_n$ takes at least as many iterations as for SWITCH-ALL.
\end{claimproof}

Note that the number of nodes and edges of $G_n$ for regular SSI is $2n+2$ and $6n$, respectively. 
Also, the running time of SSI cannot be more than the total number of strategies for both players, which is exponential in the number of nodes and edges. 
Hence the worst-case running time of symmetric strategy improvement for parity games is exponential in the number of nodes and edges. 
Moreover, we already noted that in sink parity games, the valuations from the literature are equivalent. 
Furthermore, the result for mean and discounted payoff games can be shown analogously to \cite[Theorem~4.19]{Friedmann2011ExponentialPrograms}. 
This concludes the proof of Theorem~\ref{thm:worstcase}.

Finally, the number of nodes and edges of our game $G_n$ for generalized SSI is also linear in $n$. 
This yields the following result.
\begin{theorem}
The results from Theorem \ref{thm:worstcase} also hold for generalized SSI.
\end{theorem}

\section{Discussion}
We showed that both symmetric strategy improvement (SSI) and a generalization have exponential worst-case time complexity. 
The reason is that they can make too few crucial switches, as they are distracted by their opponent's (bad) strategy. 
In our example, the opponent's valuation always steers the players to make `bad' improving moves. 
Hence, even using the local information of the other player's valuation is not always useful (an implementation of SSI and its worst-case example can be found in Oink~\cite{vanDijk2018Oink:Solvers}). 
Remarkably, no improvement rule can fix this issue. 
Furthermore, we presented a generalization of SSI and its worst-case example which allows for more flexibility. 

The parity game example presented in Figure \ref{fig:sink_main} occurs to be trivial to solve for various natural implementations. 
One could think of preprocessing techniques like SCC decomposition, removing self-loops or choosing an initial strategy with some heuristic (in fact, the construction in the proof of Lemma \ref{lem:reducetosink} gives such a heuristic that would solve the parity game quickly). 
There are further tweaks like propagating information about which nodes are won through the graph by attracting towards them. 
However, we argue that the main principles of our counterexamples are robust against such tricks. Imagine the parity game of Figure \ref{fig:sink_main} to be part of a larger parity game, where $a_{n+1}$ looks much higher valued than $d_{n+1}$ at the beginning. Suppose that in reality $d_{n+1}$ is winning for player 0 (or player 0 has a strategy such that it has a large valuation in the version of SSI used here), and $a_{n+1}$ is winning for player 1 (or player 1 can let it have a low valuation). 
Then, until SSI discovers the true value of $a_{n+1}$ and $d_{n+1}$, the best possible strategies to play from $a_1, a_2, \ldots$ and $d_1, d_2, \ldots$ are in fact the strategies we use as initial strategies in Section \ref{sec:counterexample-basic}. So most likely, the preprocessing would pick these strategies or we quickly end up with them after some iterations. 
Secondly, SCC decomposition can be tricked by adding edges back from the rest of the game to $a_1, d_1$ that are bad choices for the player that can choose them, and the propagation trick is stopped by the modules from Figure \ref{fig:relsymmod}.  It is likely that other optimizations, assuming that they have a weakness, will be stopped by adding modules that exploit this weakness in a similar way. 
Note, however, that structures like those in our hard instances are unlikely to show up in practice. 
As observed by the authors of SSI \cite{Schewe2015SymmetricImprovement}, the number of iterations of SSI is low on randomly generated instances.

Our work leaves open if there is actually any way to use the interplay of strategies of the two players which does not end up with exponential worst-case running time. 
Our constructions suggest that the restriction imposed by using local information of the opponent's strategy might always be exploited to lure the iterations into an exponential trap. 

\bibliography{references}

\appendix
\newpage
\section{Example}

\begin{example}\label{ex:reduce}
Now we illustrate the reduction to a sink parity game and the symmetric strategy improvement algorithm with an example. Suppose we have a parity game as on the left in Figure \ref{fig:symmexample}. This can be transformed into a sink parity game as in Lemma \ref{lem:reducetosink}. Two copies of the resulting sink parity game are shown on the right in Figure \ref{fig:symmexample}. One copy shows a trivial admissible player 0 strategy $\sigma$, with its optimal response $\bar{\sigma}$. The other one shows an admissible $\tau$ with its optimal response $\bar{\tau}$.

\begin{figure}[t]
    \centering
    \includegraphics[width=\linewidth]{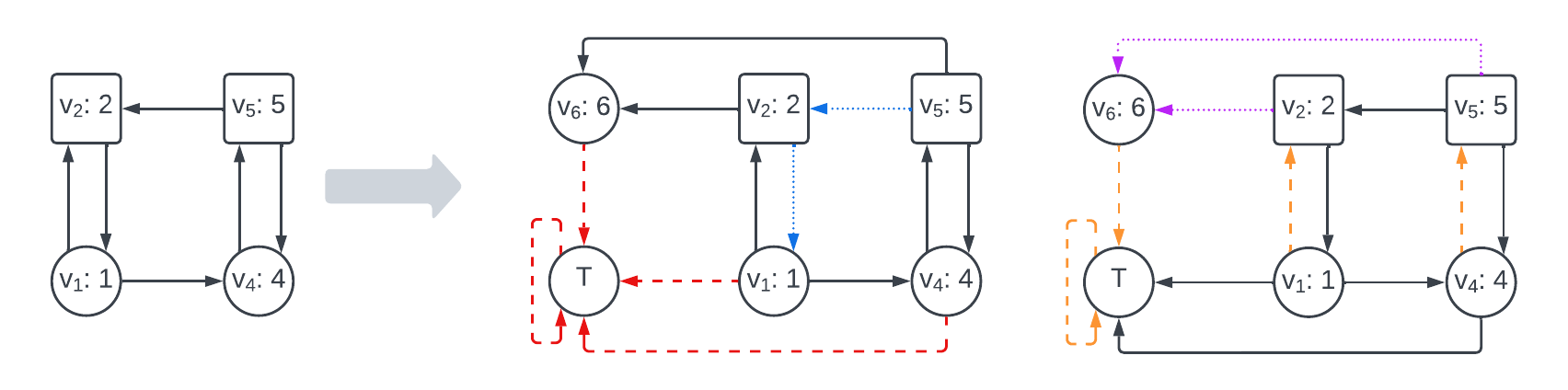}
    \caption{Left: a parity game. Middle: strategy $\sigma$ (dashed) and counterstrategy $\bar{\sigma}$ (dotted) in the resulting sink parity game. Right: Strategy $\tau$ (dashed) and counterstrategy $\bar{\tau}$ (dotted).}
    \label{fig:symmexample}
\end{figure}

Now we look at what symmetric strategy improvement does if we start with the pair of strategies shown in Figure~\ref{fig:symmexample}. 
We denote the valuations $\Xi_{\sigma}$ and $\Xi_{\tau}$ in this game by $(a_1,a_2,a_4,a_5,a_6)$, where $a_i$ stands for $\left(\Xi_{\sigma}\right)_i$ or $\left(\Xi_{\tau}\right)_i$. Now we find the improving moves. Player 0 has two improving moves, the first is $(v_1,v_2)$, as $\Xi_{\sigma}(v_2)=(1,1,0,0,0)\tr \Xi_{\sigma}(\sigma(v_1))=\Xi_{\sigma}(\top)=(0,0,0,0,0)$, and the other improving move is $(v_1,v_4)$. Since $\bar{\tau}(v_1)=v_2$, the symmetric strategy improvement algorithm will switch player 0's choice in $v_1$ to $v_2$. Player 1 has one improving move, namely $(v_5,v_4)$, as $\Xi_{\tau}(v_4)=(0,0,1,1,1)\tl (0,0,0,0,1)=\Xi_{\tau}(v_6)$. However, $\bar{\sigma}(v_5)=v_2$, so $I_{\tau}\cap \{(v,\bar{\sigma}(v))| v\in V_0\}$ is empty, and we do not change $\tau$ in this iteration. So in the first iteration, the only switch that is made is changing $\sigma(v_1)$ to $v_2$. The reader can verify that in the second iteration of the symmetric strategy improvement algorithm, only the switch $(v_5,v_4)$ is made for player 1, and that after that, the algorithm terminates. 
Then, we have a pair of optimal strategies $\sigma$ and $\tau$ as shown on the left in Figure~\ref{fig:symmexample2}. 
We can now infer the winning sets $W_0$ and $W_1$ of the original game from the $6$-element of the valuations of the nodes. Also, the optimal strategies in the sink game form winning strategies in the original game.
\begin{figure}[t]
    \centering
    \includegraphics[width=0.8\linewidth]{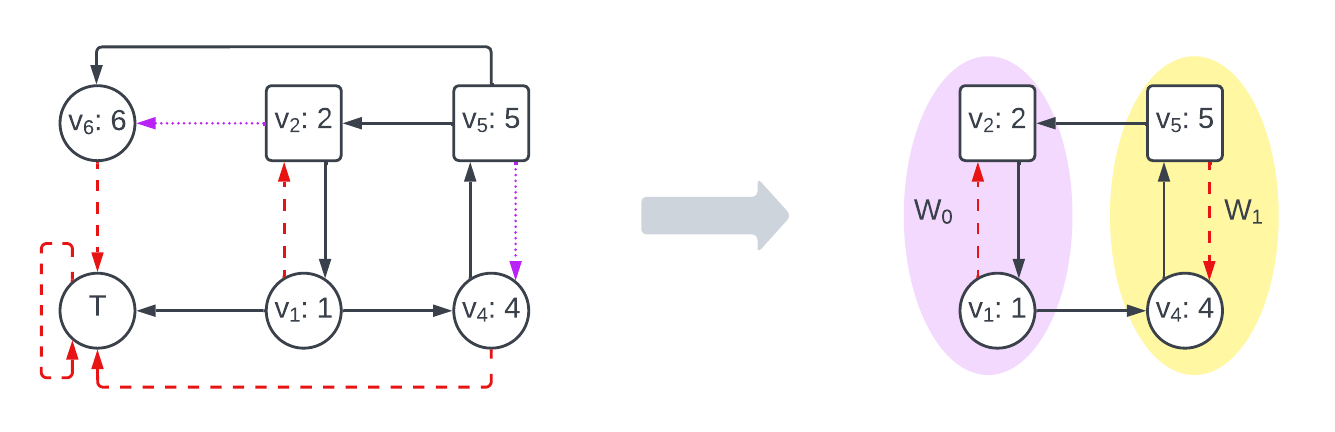}
    \caption{Left: optimal strategies $\sigma$ (dashed) and $\tau$ (dotted), found after performing symmetric strategy improvement. Right: the resulting winning sets and winning strategies (dashed) in the original game.}
    \label{fig:symmexample2}
\end{figure}
\end{example}

\section{Note on further increasing the number of switches}
On the number of switches, note that with the adapted counterexample, generalized symmetric strategy improvement still makes `too few' switches, because it puts off making some good switches. It considers the large difference in $\Xi_{\tau}$ for player 0 switches and the large difference in $\Xi_{\sigma}$ for player 1 switches. That opens the question if one could further increase the number of switches that generalized symmetric strategy improvement makes. A next logical increase would be to always make an improving move in a node if one can, and then possibly be guided by the opponent's strategy for deciding which improving move to make. However, in a run of regular strategy improvement on the switch-best counterexample from \cite{Friedmann2011ExponentialPrograms}, there is always exactly one improving move per node, and therefore this generalization would behave the same on this game for player 0 as strategy improvement where every possible switch is made. \footnote{There is one exception in the iteration after their so-called deceleration lane resets, here we can choose to which node in the deceleration lane we go. However, for this generalization, it is both from the perspective of $\Xi_{\sigma}$ and $\Xi_{\tau}$ always better to switch to the highest even node in the lane like in regular strategy improvement. Hence it would make sense to assume that this always happens.} Hence there is still an exponential-time worst-case example. This time, one could say that it is because the algorithm is switching too many edges, as the main idea behind this example is that the parts representing significant bits are distracted by unnecessary switches so that the insignificant bits are being switched very often like a binary counter.
\end{document}